\documentclass[11pt]{article}
\usepackage{geometry}
\usepackage{algorithm}
\usepackage{algorithmic}
\usepackage{amsthm}
\usepackage{amsmath}
\usepackage{amssymb}
\usepackage{epstopdf}
\usepackage{cite, bm, braket, graphicx, booktabs, float}
\usepackage[colorlinks,linkcolor=blue,citecolor=blue,anchorcolor=blue]{hyperref}
\usepackage[matrix,frame,arrow]{xypic}
\usepackage{qcircuit}
\usepackage[title]{appendix}

\title{Quantum radial basis function method for the Poisson equation}

\author{Lingxia Cui\footnotemark[1]\
     \and Zongmin Wu\footnotemark[2]
    \and Hua Xiang\footnotemark[1] \footnotemark[3] \footnotemark[4]}
\date{}

\begin{document}
    \maketitle
    \renewcommand{\thefootnote}{\fnsymbol{footnote}}
    \footnotetext[1]{School of Mathematics and Statistics, Wuhan University, Wuhan 430072, China.}
    \footnotetext[2]{School of Mathematics, Fudan University, Shanghai 200433, China.}
    \footnotetext[3]{Hubei Key Laboratory of Computational Science, Wuhan University, Wuhan 430072, China.}
    %\footnotetext[4]{Corresponding author. E-mail address: {\tt hxiang@whu.edu.cn}.}
    \footnotetext[4]{E-mail address: {\tt hxiang@whu.edu.cn}.}

\begin{abstract}
The radial basis function (RBF) method is used for the numerical solution of the Poisson problem in high dimension. The approximate solution can be found  by solving a large system of linear equations. Here we investigate the extent to which the RBF method can be accelerated using an efficient quantum algorithm for linear equations. We compare the theoretical performance of our quantum algorithm with that of a standard classical algorithm, the conjugate gradient method. We find that the quantum algorithm can achieve a polynomial speedup.
\end{abstract}

\section{Introduction}
Partial differential equations (PDEs) are ubiquitous in science and engineering. The development of efficient and accurate numerical algorithms is one of the most important field in computational mathematics discipline and engineering analysis community.
Traditional numerical methods for solving PDEs are often based on mesh discretization, such as finite difference, finite elements or finite volumes. Finite difference provides an efficient numerical approach that permits large-scale simulations with high-order accuracy in many areas, but requires structured meshes. Finite element methods are flexible to handle complicated geometry, but both coding and mesh generation become increasingly difficult when the number of space dimensions increases. Finite volume methods involve discretizing space into many grid cells, also having the mesh generation problem. An alternative to mesh-based methods is spectral methods, which usually provide high accuracy, but have severe regularity restrictions on the geometry.

When solving PDEs, it is very desirable to use entirely meshfree data distributions just as needed to fit boundaries and to satisfy spatially variable resolution requirements, but without having to form any local meshes. In many applications, it is also desirable to deal with high-dimensional problems.
Being entirely meshfree and insensitive to space dimensions, RBF discretizations are particularly easy to implement in complicated geometry domain and can be easily adaptable to problems involving more than two independent variables.

During the past decades, RBF techniques as a powerful tool have gained popular uses to solving PDEs \cite{fasshauer1999solving, franke1998convergence, franke1998solving, wendland2007stability, fornberg2015solving, wong2002compactly, wu2003convergence}. While Galerkin methods have been suggested \cite{wendland1999meshless}, the predominant method for the numerical solution of PDEs using RBFs is based upon collocation,
due to the obvious advantage that no numerical integration is required to set up the discretized linear system.
The collocation method employs linear combinations of RBFs or the derivative information to satisfy collocation conditions for discretizing PDEs. This further leads to a system of linear equations of dimension equal to the number of collocation points. Through solving such a linear system, an explicit formula of the numerical solution of PDEs is essentially attained.
Since the computational cost and ill-conditioning of using global RBFs are major factors to be considered, when dealing with large-scale problems, a local scheme with compactly supported RBFs (CSRBFs) can be more useful.
The advantages are a sparse collocation matrix, and the possibility of a fast evaluation of the numerical solution.
Though equipped with a sparse collocation matrix by CSRBFs, any classical algorithm requires time that scales at least as the number of collocation points for solving the resultant linear system. For problems with a large number of points in high space dimensions, the computational cost can be an obvious concern.

Recently, quantum algorithms have received much attention because of the
enhanced computational complexity and potential applications in many different areas \cite{shor1994algorithms,grover1997quantum,harrow2009quantum,wiebe2012quantum,wang2019quantum,xiang2022quantum,rebentrost2014quantum, biamonte2017quantum, berry2014high, berry2017quantum}. One area where quantum algorithms could offer significant speedups over the classical counterparts is the numerical solution of PDEs. Indeed, the quantum linear system algorithm (QLSA) forms a foundation and lies at the heart of many quantum algorithms \cite{cao2013quantum, montanaro2016quantum, childs2021high, chen2022quantum}. The reason is that many classical approaches (e.g. the aforementioned finite difference, finite element, finite volume, spectral or RBF methods) for solving PDEs are based on discretizing the PDE and reducing the problem to a linear system. Beginning with HHL algorithm \cite{harrow2009quantum}, various QLSAs \cite{childs2017quantum, wossnig2018quantum, subacsi2019quantum, lin2020optimal} have been proposed and applied to the numerical solution of PDEs. We just list a few examples. Combining the finite difference method with the HHL algorithm, the authors in \cite{cao2013quantum} gave a quantum algorithm for the Poisson equation with improved performance. By applying the LCU-based QLSA \cite{childs2017quantum} to the finite element method, the authors in \cite{montanaro2016quantum} presented a quantum algorithm for the Poisson equation which achieves a polynomial speedup over the classical ones. The authors in \cite{chen2022quantum} proposed a quantum approach to accelerate the finite volume method for steady computational fluid dynamics problems built on the LCU-based quantum linear solver.
Combining the spectral method with the LCU-based QLSA, the authors in \cite{childs2021high} developed a high-precision quantum algorithm for second-order elliptic equations with an exponential speedup over the classical counterparts. Besides, allowing for easy to implement on noisy intermediate-scale quantum devices, a variational quantum algorithm to solve the Poisson equation was proposed in \cite{liu2021variational}.

Even though the wide investigation of quantum algorithms for the numerical solution of PDEs, the quantum RBF approaches to approximating the solution of PDEs have not been considered to the knowledge of the authors. It may be very significant due to the geometric flexibility and extensive applications of  RBF methods to solving PDEs especially for high-dimensional problems.
Here we investigate a fundamental problem, the Poisson equation that arises in numerous areas of science and engineering, to enable us to compare the complexities of classical and quantum algorithms. By means of the technique of QLSA based on filtering \cite{lin2020optimal} and block encodings \cite{gilyen2018quantum},
we develop a quantum version of the collocation method by CSRBFs to deal with the $d$-dimensional Poisson problem with Dirichlet boundary conditions.
We further implement an overall complexity analysis.
To achieve accuracy $\epsilon$, any classical numerical method requires cost bounded by a function that grows as $\epsilon^{-O(d)}$ \cite{werschulz1991computational}. However, the scaling with $\epsilon$ of our quantum algorithm's runtime does not grow exponentially in the space dimension $d$. The quantum algorithm can achieve a polynomial speedup for high-dimensional problems compared to the classical counterparts.

Certain items of notations are used throughout the paper. We employ $||\cdot||$ to denote vector or matrix $2$-norm for notational convenience. We need to deal with continuous functions, their discretized approximations as vectors, and their quantum states. Italics denote functions, boldface denotes vectors, and quantum states (usually normalized) are represented as kets. We use $s$ to describe the sparsity of a matrix if there are no more than $s$ nonzero entries in any row or column of the matrix.
Sometimes we will use the same notation to describe a kind of object if it is obvious from the context. For example, we may use $C$ to represent the constant without confusion. The $\widetilde{O}$ notation is used to suppress more slowly growing factors.

The paper is organized as follows. Section 2 introduces technical details about CSRBFs and collocation methods, and formally states the problem we solve. In Section 3 we discuss the quantum RBF algorithm for the Poisson problem and perform an overall error analysis. Section 4 concludes with a brief discussion of the results and lists the possible future research.

\section{Collocation method based on CSRBFs}
We focus on the collocation method by CSRBFs to solve the Poisson problem. In this section, some related technical details are introduced, and the problem that we will solve by quantum algorithms is formally stated.
For a radial function $\Phi(\bm x)=\phi(||\bm x||)$ for $\bm x\in\mathbb{R}^d$, the capital letter $\Phi$ denotes the multivariate function and the small letter $\phi$ denotes the univariate function. For a set of distinct collocation points
$\mathcal{X}=\{\bm x_1,\dots,\bm x_N\}$ in a bounded domain $\Omega\subseteq\mathbb{R}^d$, the fill distance, also called the mesh norm, is defined as
\begin{equation}\label{eq.2.1}
    h=h_{\mathcal{X},\Omega}=\sup\limits_{\bm x\in\Omega}\min\limits_{1\leq j\leq N}||\bm x-\bm x_j||,
\end{equation}
which measures the radius of the largest data-free ball contained in $\Omega$.
The separation distance is defined by
\begin{equation}\label{eq.2.2}
    q=q_\mathcal{X}=\frac{1}{2}\min\limits_{i\neq j}||\bm x_i-\bm x_j||,
\end{equation}
which is half of the smallest distance between two points in $\mathcal X$.
The data set $\mathcal{X}$ is said to be quasi-uniform (such as Halton points \cite{halton1960on}) with respect to a constant $c_{qu}>0$ if
\begin{equation}\label{eq.2.3}
    q\leq h\leq c_{qu}q,
\end{equation}
i.e., the separation distance and the fill distance are of comparable size.

\subsection{CSRBFs}
A family of CSRBFs was first introduced by Wu in the mid 1990s \cite{wu1995compactly}. Starting with a cutoff polynomial, then using convolution and a differential operator, Wu constructed a series of positive definite radial functions with compact support, for any given dimension and prescribed smoothness. Later, Wendland expanded Wu's results, and constructed a family of positive definite CSRBFs of minimal degree for given smoothness and space dimension \cite{wendland1995piecewise}, which are usually called Wendland's functions. There are also other ways to construct positive definite CSRBFs \cite{gneiting2002compactly}. In this paper, we work with Wendland's functions for simplicity.

Wendland's functions $\Phi_{d,k}=\phi_{d,k}(||\cdot||)\in C^{2k}(\mathbb{R}^d)$ take the form
\begin{equation}\label{eq.2.1.4}
\phi_{d,k}(r) = (1-r)_+^\ell p(r)
\end{equation}
with the following conditions
\begin{equation*}
    (1-r)_+^\ell= \left\{
    \begin{array}{ll}
        (1-r)^\ell, & 0\leq r\leq1,\\
        0, & r>1,
    \end{array}
    \right.
\end{equation*}
where $r=||\bm x-\bm x_j||$ is the Euclidean distance between the input variable $\bm x$ and the center $\bm x_j$, $\ell=\lfloor d/2\rfloor+k+1$, and $p(r)$ is a polynomial \cite{wendland2004scattered}. Sometimes, the subscript $d,k$ is omitted without confusion. We assume the collocation points and the centers coincide without loss of generality.

These radial functions $\phi_{d,k}$ are positive definite on $\mathbb{R}^d$ with degree of smoothness $2k$, and have support radius equal to 1.  Table \ref{table.1} lists some explicit formulas of Wendland's functions, where $\doteq$ denotes equality up to a multiplicative positive constant.
Scaling of the radial functions by replacing $r$ with $r/\delta$ for $\delta>0$, i.e., $\phi_{d,k}(r/\delta)$, allows any desired support radius $\delta$. In general, the smaller the value of $\delta$ is, the sparser the associated matrix becomes; however this also results in lower accuracy. That is, a tradeoff principle exists between the computational efficiency and the accuracy. The choice of the optimal support radius $\delta$ of a CSRBF is a delicate question.

\begin{table}[H]
    \centering
    \caption{\label{table.1} Explicit formulas of Wendland's functions $\phi_{d,k}$ for $k=0,1,2,3$. $\ell=\lfloor d/2\rfloor+k+1$.}
    \renewcommand\arraystretch{1.3}
    \begin{tabular}{l c}
        \toprule
        \textbf{Function} & \textbf{Smoothness}\\
        \midrule
        $\phi_{d,0}(r)=(1-r)_+^{\lfloor d/2 \rfloor + 1}$ & $C^0$ \\
        $\phi_{d,1}(r)\doteq(1-r)_+^{\ell + 1}\left[(\ell+1)r+1\right]$ & $C^2$ \\
        $\phi_{d,2}(r)\doteq(1-r)_+^{\ell + 2}\left[ (\ell^2+4\ell+3)r^2 +(3\ell+6)r +3 \right]$ & $C^4$ \\
        $\phi_{d,3}(r)\doteq(1-r)_+^{\ell + 3}\left[ (\ell^3+9\ell^2+23\ell+15)r^3 +  (6\ell^2+36\ell+45)r^2 +(15\ell+45)r +15 \right]$ & $C^6$\\
        \bottomrule
    \end{tabular}
\end{table}

Recall the definition of a positive definite radial function and its elementary properties.

\newtheorem{definition}{Definition}
\begin{definition}{\rm \cite{fasshauer2007meshfree}}
A continuous function $\Phi: \mathbb{R}^d\rightarrow\mathbb{C}$ is called positive definite if, for all $N\in\mathbb{N}$, all sets of pairwise distinct points $\mathcal{X}=\{\bm x_1,\dots,\bm x_N\}\subseteq\mathbb{R}^d$, and all $\bm \xi\in\mathbb{C}^N$, the quadratic form
\begin{equation*}
    \sum\limits_{i=1}^N\sum\limits_{j=1}^N\xi_i\overline{\xi_j}\Phi(\bm x_i-\bm x_j)
\end{equation*}
is positive for all $\bm\xi\in\mathbb{C}^N\backslash \bm0$. We call a univariate function $\phi: [0,\infty)\rightarrow\mathbb{R}$ positive definite on $\mathbb{R}^d$ if the corresponding multivariate function $\Phi(\bm x)=\phi(||\bm x||)$, $\bm x\in\mathbb{R}^d$, is positive definite.
\end{definition}

\newtheorem{theorem}{Theorem}
\begin{theorem} {\rm \cite{fasshauer2007meshfree}} \label{th.01}
{\rm (I)} Suppose $\Phi$ is a positive definite radial function on $\mathbb{R}^d$. Then
$\Phi(\bm 0)\geq0$ and $|\Phi(\bm x)|\leq\Phi(\bm 0)$ for all $\bm x\in\mathbb{R}^d$.\\
{\rm (II)} Let $\Phi$ be a continuous function in $L_1(\mathbb{R}^d)$. $\Phi$ is positive definite if and only if $\Phi$ is bounded, and its Fourier transform is non-negative and not identically equal to zero.
\end{theorem}

Theorem \ref{th.01} shows that the positive definiteness of a radial function $\Phi$ on $\mathbb{R}^d$ relies on its Fourier transform
\begin{equation}\label{eq.2.1.5}
    \widehat{\Phi}(\bm\omega)=\frac{1}{(2\pi)^{d/2}}\int_{\mathbb{R}^d}\Phi(\bm x)e^{-{\mathbf i} \bm x^T \bm\omega}d\bm x,\quad \bm\omega\in\mathbb{R}^d,
\end{equation}
where ${\mathbf i}$ denotes the imaginary unit.
Every positive definite radial function $\Phi$ can indeed generate a reproducing kernel Hilbert space, also called its native space $\mathcal N_\Phi(\mathbb{R}^d)$. Suppose that the Fourier transform of $\Phi$ satisfies the algebraic decay condition
\begin{equation}\label{eq.2.1.6}
    c_1(1+||\bm \omega||^2)^{-\tau}\leq\widehat{\Phi}(\bm\omega)\leq c_2(1+||\bm \omega||^2)^{-\tau},\quad \bm\omega\in\mathbb{R}^d
\end{equation}
with $\tau>d/2$ and two fixed constants $0< c_1\leq c_2$. Then the native space $\mathcal N_\Phi(\mathbb{R}^d)$ is norm equivalent to the Sobolev space $H^\tau(\mathbb{R}^d)$ \cite{wendland2004scattered}. If $\tau>k+d/2$ then Fourier inversion formula guarantees that $\Phi\in C^{2k}(\mathbb{R}^d)$ and Sobolev's embedding theorem shows that $H^\tau(\mathbb{R}^d)\subseteq C^{k}(\mathbb{R}^d)$. Examples of such functions are Wendland's functions $\Phi_{d,k}$ in (\ref{eq.2.1.4}), the native space of which is a classical Sobolev spaces, i.e. $\mathcal N_{\Phi_{d,k}}(\mathbb{R}^d)=H^{\tau} (\mathbb{R}^d)$ with $\tau=d/2+k+1/2$.

We will scale the radial function in the following way. Let $\Phi_\delta$ be defined by
\begin{equation}\label{eq.2.1.7}
\Phi_\delta(\bm x) := \delta^{-d}\Phi(\bm x/\delta)
\end{equation}
such that it has a Fourier transform $\widehat{\Phi_\delta}(\bm\omega)=\widehat{\Phi}(\delta\bm\omega)$. Note that $\Phi$ has support radius $1$, and $\Phi_\delta$ has support radius $\delta$.
Using the scaled basis function $\Phi_\delta$ yields the following norm equivalence.
\newtheorem{lemma}{Lemma}
\begin{lemma}{\rm \cite{wendland2010multiscale}} \label{lemma.01}
For every $\delta\in(0,1]$ we have $\mathcal{N}_{\Phi_\delta}(\mathbb{R}^d)=H^{\tau} (\mathbb{R}^d)$ and for every $y\in H^\tau(\mathbb{R}^d)$, we have the norm equivalence
\begin{equation*}
    c_1^{1/2}||y||_{\Phi_\delta}\leq||y||_{H^\tau(\mathbb{R}^d)}\leq c_2^{1/2}\delta^{-\tau}||y||_{\Phi_\delta},
\end{equation*}
with $c_1, c_2>0$ constants in (\ref{eq.2.1.6}) and $||\cdot||_{\Phi_\delta}$ the norm of the native space $\mathcal{N}_{\Phi_\delta}(\mathbb{R}^d)$. The condition $\delta\leq1$ can be relaxed to $\delta\leq\delta_0$ for any $\delta_0>0$ at the price of having constants depending on $\delta_0$.
\end{lemma}

\subsection{Collocation method}
Collocation plays a key role in the RBF approaches to solving PDEs. There are two main approaches taken when formulating the RBF expansion for the collocation solution of PDEs. One is a nonsymmetric method suggested by Kansa \cite{kansa1990multiquadrics}. This method has the advantage that less derivatives have to be formed but has the drawback of a nonsymmetric coefficient matrix, which might even become singular for certain configurations of centers \cite{hon2001unsymmetric}. The other one is a symmetric collocation method, motivated by scattered Hermite-Birkhoff interpolation developed by Wu \cite{zongmin1992hermite}. This Hermite-based method produces a symmetric coefficient matrix that will be nonsingular as long as the radial function is chosen appropriately.

In this paper, we concentrate on the symmetric collocation method by CSRBFs to deal with a $d$-dimensional Poisson problem, defined on a bounded domain $\Omega\subseteq\mathbb{R}^d$ with Dirichlet boundary conditions as
\begin{equation}\label{eq.2.2.8}
    \begin{aligned}
    -\Delta u(\bm x) &= f(\bm x), &\bm x&\in\Omega, \\
    u(\bm x) &= g(\bm x), &\bm x&\in\partial\Omega,
    \end{aligned}
\end{equation}
where $\Delta=\sum_{j=1}^{d}\partial^2/\partial x_j^2$ is the Laplace operator with respect to $\bm x=(x_1,\dots,x_d)\in\mathbb{R}^d$.
The problem is discretized on two sets $\mathcal{X}=\mathcal{I}\cup\mathcal{B}$ of $N$ distinct collocation points, where
$\mathcal{I}=\{\bm x_1,\dots,\bm x_{N_\mathcal{I}}\}$ contains $N_\mathcal{I}$ interior points and $\mathcal{B}=\{\bm x_{N_\mathcal{I}+1},\dots,\bm x_N\}$ contains $N-N_\mathcal{I}$ boundary points.
The collocation (numerical) solution of the Poisson equation can be represented as
\begin{equation}\label{eq.2.2.9}
    \bar u(\bm x) = -\sum\limits_{j=1}^{N_\mathcal{I}}c_j\Delta\Phi_\delta(\bm x-\bm x_j) + \sum\limits_{j=N_\mathcal{I}+1}^Nc_j\Phi_\delta(\bm x-\bm x_j),
\end{equation}
where $c_j$ for $j=1,\dots,N$ are expansion coefficients, and $\Phi_\delta=\delta^{-d}\Phi_{d,k}(\cdot/\delta)$ is a scaled Wendland's function with support radius $\delta$ as defined in (\ref{eq.2.1.7}). Note that the native space generated by Wendland's function $\Phi_{d,k}$ is a Sobolev space $H^\tau(\mathbb{R}^d)$ with
\begin{equation}\label{eq.2.2.10}
\tau=d/2+k+1/2.
\end{equation}

After enforcing the collocation conditions
\begin{equation*}
    \begin{aligned}
        -\Delta \bar u(\bm x_j) & = f(\bm x_j), & \bm x_j\in\mathcal{I}, \\
        \bar u(\bm x_j) &= g(\bm x_j), & \bm x_j \in\mathcal{B},
    \end{aligned}
\end{equation*}
we end up with a linear system
\begin{equation}\label{eq.2.2.11}
A^{\diamond}\bm c^{\diamond} = \bm b^{\diamond},
\end{equation}
where $A^{\diamond}$ is the collocation matrix that is symmetric positive definite in the form
\begin{equation}\label{eq.2.2.12}
    A^{\diamond} = \left[
    \begin{array}{cc}
        A_\mathcal{II} & A_\mathcal{IB} \\
        A_\mathcal{IB}^T & A_\mathcal{BB}
    \end{array}
    \right]
\end{equation}
with the entries of each block given by
\begin{equation*}
    \begin{aligned}
        (A_\mathcal{II})_{ij} & = \Delta^2\Phi_\delta(\bm x_i-\bm x_j), & \bm x_i&,\bm x_j \in \mathcal{I}, \\
        (A_\mathcal{IB})_{ij} & = -\Delta\Phi_\delta(\bm x_i-\bm x_j), & \bm x_i&\in \mathcal{I}, \bm x_j \in \mathcal{B},\\
        (A_\mathcal{BB})_{ij} & = \Phi_\delta(\bm x_i-\bm x_j), & \bm x_i&,\bm x_j \in \mathcal{B},\\
    \end{aligned}
\end{equation*}
and
\begin{equation}\label{eq.2.2.13}
\bm c^{\diamond}=\left[c_1,\cdots,c_N\right]^T, \quad \bm b^{\diamond} =\left[f(\bm x_1),\dots,f(\bm x_{N_\mathcal{I}}),g(\bm x_{N_\mathcal{I}+1}),\dots,g(\bm x_N)\right]^T.
\end{equation}
By solving such a linear system, an explicit formula of $\bar u(\bm x)$ in (\ref{eq.2.2.9}) can be determined. As for the convergence of the numerical solution, we have the following result from \cite{farrell2013rbf}.

\begin{theorem} \label{th.02}
Assume that $\delta\in(0,1]$. Let $u\in H^\tau(\Omega)$ be the solution of
the Poisson equation (\ref{eq.2.2.8}). Let the domain $\Omega$ have a $C^{\alpha,\gamma}$-boundary for $\gamma\in[0,1)$ such that $\tau=\alpha+\gamma$ and $\alpha := \lfloor\tau\rfloor>2+d/2 $. Then the error between the solution $u$ and its collocation solution $\bar u$ in (\ref{eq.2.2.9}) can be bounded by
\begin{equation*}
    ||u-\bar u||_{L_2(\Omega)}\leq C\delta^{-\tau}h^{\tau-2}||u||_{H^\tau(\Omega)}
\end{equation*}
in the $L_2$-norm for sufficiently small $h$, where $h$ is the larger of the fill distances in the interior and on the boundary of $\Omega$.
\end{theorem}

In the end, take a closer look at the entries of the collocation matrix $A^{\diamond}$ in (\ref{eq.2.2.12}). The following lemma can be derived with a detailed proof in Appendix A.

\begin{lemma}\label{lemma.02}
Let $\Phi(\bm x)=\phi(r)$ with support radius $1$ and $r=||\bm x||$ for $\bm x=(x_1,\dots,x_d)\in\mathbb{R}^d$. Suppose $\Phi\in C^4(\mathbb{R}^d)$. Denote by $\Delta$ the Laplace operator $\Delta=\sum_{j=1}^{d}\partial^2/\partial x_j^2$.
Then we have
\begin{equation}\label{eq.2.2.14}
    \begin{split}
    \Delta\Phi(\bm x) &= \phi''(r)+\frac{d-1}{r}\phi'(r) := F_1(r), \\
    \Delta^2\Phi(\bm x) &= \phi^{(4)}(r)+\frac{2(d-1)}{r}\phi'''(r)+\frac{(d-1)(d-3)}{r^2}\phi''(r)-\frac{(d-1)(d-3)}{r^3}\phi'(r) := F_2(r),
    \end{split}
\end{equation}
provided that $r\neq0$ or when $r=0$ these two formulas are well defined.
$\Delta^2$ represents the double Laplacian, and $F_1$, $F_2$ are radial functions with support radius $1$.
\end{lemma}

\newtheorem{remark}{Remark}
\begin{remark}\label{rem.01}
If $\Phi\in C^4(\mathbb{R}^d)$ takes the form of Wendland's functions, the coefficients of $r$ and $r^3$ in the univariate polynomials $\phi$ vanish by \cite[Theorem 9.12]{wendland2004scattered}. It is easy to verify that $\Delta\Phi$ and $\Delta^2\Phi$ in (\ref{eq.2.2.14}) are well defined when $r=0$; besides, they are also radial functions with compact support. By the positive definiteness of $\Phi$, we can also infer that $-\Delta\Phi$ and $\Delta^2\Phi$ are positive definite based on their Fourier transforms that can be derived from \cite{stein1971introduction} and Theorem \ref{th.01}. In addition, by the definition of $\Phi_\delta$ in (\ref{eq.2.1.7}), we have
\begin{equation*}
    \Delta^{2}\Phi_\delta(\bm x) = \delta^{-d-4}F_1(||\bm x||/\delta), \quad \Delta\Phi_\delta(\bm x) = \delta^{-d-2}F_2(||\bm x||/\delta).
\end{equation*}
\end{remark}

By Remark \ref{rem.01}, the collocation matrix $A^{\diamond}$ from (\ref{eq.2.2.12}) can be further rewritten as
\begin{equation}\label{eq.2.2.15}
    A^{\diamond}
    = \delta^{-d}\left[
    \begin{array}{cc}
        \delta^{-4}\mathcal{F}_\mathcal{II} & -\delta^{-2}\mathcal{F}_\mathcal{IB} \\
        -\delta^{-2}\mathcal{F}_\mathcal{IB}^T & \mathcal{F}_\mathcal{BB}
    \end{array}
    \right],
\end{equation}
where $\delta^{-4}$ and $\delta^{-2}$ are factors of the block matrices, and
\begin{equation}\label{eq.2.2.16}
    \begin{aligned}
        (\mathcal{F}_\mathcal{II})_{ij} & = F_1(r_{ij}/\delta), & \bm x_i&,\bm x_j \in \mathcal{I}, \\
        (\mathcal{F}_\mathcal{IB})_{ij} & = F_2(r_{ij}/\delta), & \bm x_i&\in \mathcal{I}, \bm x_j \in \mathcal{B},\\
        (\mathcal{F}_\mathcal{BB})_{ij} & = \phi(r_{ij}/\delta), & \bm x_i&,\bm x_j \in \mathcal{B}
    \end{aligned}
\end{equation}
with $r_{ij} := ||\bm x_i-\bm x_j||$ the Euclidean distance and $F_1, F_2$ the radial functions defined in (\ref{eq.2.2.14}).

\subsection{Diagonal preconditioner}
To ensure computational efficiency, it is desirable to use small support radius $\delta$ for a sparse collocation matrix. However, small values of $\delta$ will induce ill-conditioned collocation matrix. Fasshauer has numerically illustrated this, and pointed out that the reason is due to the different scaling of the different parts of the collocation matrix with respect to the support radius $\delta$ \cite{fasshauer1999solving}. It can be seen from $A^{\diamond}$ in (\ref{eq.2.2.15}) that the diagonal block corresponding to the inner points scales like $O(\delta^{-4})$, the diagonal part corresponding to the boundary points scales like $O(1)$, and the off-diagonal block scales like $O(\delta^{-2})$.

To overcome this kind of ill-conditioning, Fasshauer also provided a diagonal preconditioning strategy, and numerically verified its efficiency \cite{fasshauer1999solving}. Later, Farrell and Wendland \cite{farrell2013rbf} theoretically demonstrated this. The idea applied to our problem can be stated as follows. Instead of solving the linear system $A^{\diamond}{\bm c^{\diamond}}=\bm b^{\diamond}$ in (\ref{eq.2.2.11}), we rewrite the linear system as
\begin{equation*}
    \mathcal PA^{\diamond}\mathcal P\mathcal P^{-1}\bm c^{\diamond} = \mathcal P\bm b^{\diamond}
\end{equation*}
where
\begin{equation}\label{eq.2.3.19}
    \mathcal P={\rm diag}(\delta^2,\dots,\delta^2,1,\dots,1)
\end{equation}
is an $N\times N$ diagonal matrix with the first $N_\mathcal I$ diagonal terms $\delta^2$ and the remaining $1$.
And solve the preconditioned system as
\begin{equation}\label{eq.2.3.17}
    A \bm c = \bm b,
\end{equation}
where the preconditioned collocation matrix $A$ has the representation
\begin{equation}\label{eq.2.3.20}
A=\mathcal PA^{\diamond}\mathcal P
= \delta^{-d}\left[
\begin{array}{cc}
    \mathcal{F}_\mathcal{II} & \mathcal{F}_\mathcal{IB} \\
    \mathcal{F}_\mathcal{IB}^T & \mathcal{F}_\mathcal{BB}
\end{array}
\right]
\end{equation}
by (\ref{eq.2.2.15}) with the block matrices defined in (\ref{eq.2.2.16}), and
\begin{equation}\label{eq.2.3.18}
    \bm c = \mathcal P^{-1}\bm c^{\diamond}, \quad \bm b=\mathcal P\bm b^{\diamond}.
\end{equation}
The condition number $\kappa$ of $A$ can be bounded \cite{wendland2007stability, farrell2013rbf}.

\begin{theorem}\label{th.03}
Assume the support radius $\delta\in(0,1]$. The condition number $\kappa$ of the preconditioned collocation matrix $A$ can be bounded by
\begin{equation*}
\kappa \leq C \left(1+\delta/q\right)^d  \left(\delta/q\right)^{2\tau-d},
\end{equation*}
and the sparsity $s$ of $A$ can be bounded by $s \leq \left(1+\delta/q\right)^d$,
where $C>0$ is a constant independent of the data set $\mathcal{X}$, $q$ is the separation distance defined in (\ref{eq.2.2}), and $\tau$ indicates the order  of smoothness of the Sobolev space $H^\tau(\mathbb{R}^d)$ generated by $\Phi$ in (\ref{eq.2.2.10}).
\end{theorem}

In this paper, we aim to find a quantum RBF method solving the Poisson equation (\ref{eq.2.2.8}), and investigate the extent to which the RBF method can be accelerated. We focus on the following problem based on the symmetric collocation method by CSRBFs.

\newtheorem{problem}{Problem}
\begin{problem}\label{problem 1}
Suppose that the Poisson equation (\ref{eq.2.2.8}) has a unique solution $u\in H^\tau(\Omega)$. Let $\Omega\subseteq\mathbb{R}^d$ have a $C^{\alpha,\gamma}$-boundary for $\gamma\in[0,1)$ such that $\tau=\alpha+\gamma$ and $\alpha = \lfloor\tau\rfloor>2+d/2 $. Let  $\Phi_\delta$ be a scaled Wendland's function with support radius $\delta\in(0,1]$ defined by (\ref{eq.2.1.7}).
Suppose the collocation points in $\mathcal X$ have a quasi-uniform distribution satisfying (\ref{eq.2.3}).
In the quantum setting, assume oracles that have access to the nonzero entries of the preconditioned collocation matrix $A$, and prepare the quantum state of $\bm b$. The goal is to output a quantum state whose amplitudes are proportional to $u(\bm x)$ on the set $\mathcal X$.
\end{problem}

\section{Quantum RBF method}
The key step towards handling the Poisson problem more quickly on a quantum computer is to replace the classical algorithm for solving the preconditioned linear system (\ref{eq.2.3.17}) with a quantum algorithm.
Such a quantum algorithm was recently presented by Lin and Tong \cite{lin2020optimal} based on quantum filtering, which achieves the near optimal query complexity $\widetilde{O}(s\kappa\log(1/\epsilon_L))$ for an $s$-sparse matrix, with $\kappa$ the condition number and $\epsilon_L$ the error for the linear system. In this paper we exploit this strategy to solve (\ref{eq.2.3.17}). We further apply block-encoding techniques \cite{gilyen2018quantum} to output a quantum state encoding the solution of the Poisson equation at the collocation points.

\subsection{Block encodings}
An important idea behind the quantum RBF method is the use of block encodings. A block encoding embeds a properly scaled matrix of interest into a larger unitary operator that can be efficiently implemented on a quantum computer. An $(a+n)$-qubit unitary $U_A$ is called an $(\eta,a,\varepsilon)$-block-encoding of an $n$-qubit operator $A$ if
\begin{equation*}
    ||A-\eta(\bra{0^a}\otimes I_n)U_A(\ket{0^a}\otimes I_n)||\leq \varepsilon,
\end{equation*}
where $I_n$ denotes an $n$-qubit identity.
If we are provided a block-encoding of a Hermitian matrix $A$, quantum eigenvalue transforms based on quantum signal processing (QSP) allow to construct a block-encoding of an arbitrary polynomial transform of the eigenvalues of $A$ \cite{gilyen2018quantum, martyn2021grand}. When the polynomial naturally has a parity, the polynomial eigenvalue transformation can be specialized as follows.

\begin{theorem}(\rm Polynomial eigenvalue transformation \cite{gilyen2018quantum})\label{lemma.05}
Let $U_A$ be an $(\eta,a,\varepsilon)$-block-encoding of a Hermitian matrix $A$. Let $p\in\mathbb{R}[x]$ be a degree-$\ell$ even or odd real polynomial and $|p(x)|\leq1$ for any $x\in[-1,1]$. Then there exists a $(1,a+1,4\ell\sqrt{\varepsilon/\eta})$-block-encoding of $p(A/\eta)$ using $\ell$ queries of $U_A$, $U_A^{\dagger}$, and $O\left((a+1)\ell\right)$ other one- and two-qubit gates.
\end{theorem}

Without loss of generality, by appropriate scaling of the preconditioned collocation matrix $A$ in (\ref{eq.2.3.20}), we assume its eigenvalues are restricted to the range $D_\kappa := [-1,-1/\kappa]\cup[1/\kappa,1]$ such that $||A||\leq 1$, where $\kappa$ is the condition number of $A$. This also implies that all entries of $A$ satisfy $|A_{ij}|\leq 1$ \cite{golub2013matrix}.
We also assume that the number of the collocation points is $N=2^n$, and the sparsity of $A$ is $s=2^m$ for some integer $m<n$.
To embed $A$ into a block of a unitary matrix, assume an oracle $\mathcal{O}_{A_1}$ to locate the nonzero entries as
\begin{equation}\label{eq.3.1.21}
    \mathcal{O}_{A_1} \ket{l,j} = \ket{\nu(j,l),j}
\end{equation}
for $j\in\{1,\dots,N\}$ and $l\in\{1,\dots,s\}$, where $\nu(j,l)$ is a function that returns the row index of the $l$th nonzero entry of the $j$th column.
Assume an oracle $\mathcal{O}_{A_2}$ to compute the nonzero entries $A_{ij}$, that is,
\begin{equation}\label{eq.3.1.22}
    \mathcal{O}_{A_2} \ket{i, j, z} = \ket{i,j, z\oplus A_{ij}}
\end{equation}
for any $i,j\in\{1,\dots,N\}$, where $A_{ij}$ is the binary representation of the $(i,j)$th element. Then a block embedding of $A$ can be achieved. The strategy has also been used in \cite{childs2017quantum, gilyen2018quantum, berry2015hamiltonian}, and a variant version is given as follows.

\begin{theorem}\label{th.05}
For the $s$-sparse preconditioned collocation matrix $A$ in (\ref{eq.2.3.20}), assume $A$ has been rescaled such that its eigenvalues lie in the range $D_\kappa$. Assume we have oracles $\mathcal{O}_{A_1}$ and $\mathcal{O}_{A_2}$ defined by (\ref{eq.3.1.21}) and (\ref{eq.3.1.22}). Then we can implement an $(s,m+1,0)$-block-encoding $U_A$ of $A$ by using $O(1)$ queries to $\mathcal{O}_{A_1}$ and $\mathcal{O}_{A_2}$ in time $O({\rm poly}\log N)$, where $m=\log (s)$.
\end{theorem}
\begin{proof}
The goal is to construct $U_A$ such that $\bra{0}\bra{0^m}\bra{i}U_A\ket{0}\ket{0^m}\ket{j} = A_{ij}/s$.
Starting from $\ket{0}\ket{0^m}\ket{j}$, perform $m$ Hadamard gates $H^{\otimes m}$ on the second register; then apply the oracle $\mathcal{O}_{A_1}$ to the second and third registers to derive
\begin{equation*}
\ket{0}\ket{0^m}\ket{j} \stackrel{H^{\otimes m}}{\longrightarrow}
\frac{1}{\sqrt{s}}\sum\limits_{l=0}^{s-1}\ket{0}\ket{l}\ket{j}
\stackrel{\mathcal{O}_{A_1}}{\longrightarrow}
\frac{1}{\sqrt{s}}\sum\limits_{l=0}^{s-1}\ket{0}\ket{\nu(j,l)}\ket{j}.
\end{equation*}
Adding an ancilla register and a call to the oracle $\mathcal{O}_{A_2}$ then gives the value of $A_{\nu(j,l),j}$ in the ancilla register. Based on the value of $A_{\nu(j,l),j}$, perform a controlled rotation $R$ such that the first register is rotated from $\ket{0}$ to
\begin{equation*}
A_{\nu(j,l),j}\ket{0} + \sqrt{1-|A_{\nu(j,l),j}|^2} \ket{1}.
\end{equation*}
Then invert the oracle $\mathcal{O}_{A_2}$ to erase the value of $A_{\nu(j,l),j}$ from the ancilla register. In summary, the above steps denoted by $U_R$ act as
\begin{equation}\label{eq.3.1.23}
U_R: \ket{0}\ket{0^m}\ket{j} \rightarrow \frac{1}{\sqrt{s}} \sum\limits_{l=0}^{s-1} \left(A_{\nu(j,l),j}\ket{0} + \sqrt{1-|A_{\nu(j,l),j}|^2} \ket{1}\right) \ket{\nu(j,l)} \ket{j}.
\end{equation}
Next, define a unitary $U_L := \left(I_1\otimes SWAP\right) \left(I_1\otimes\mathcal{O}_{A_1}\right) \left(I_1\otimes H^{\otimes m}\otimes I_n\right)$, and apply it to $\ket{0}\ket{0^m}\ket{i}$ to get
\begin{equation}\label{eq.3.1.24}
    U_L: \ket{0}\ket{0^m}\ket{i} \rightarrow  \frac{1}{\sqrt{s}}\sum\limits_{l'=0}^{s-1}\ket{0}\ket{i}\ket{\nu(i,l')},
\end{equation}
where the swap operator acts on and swaps the last two registers. Take the inner product between (\ref{eq.3.1.23}) and $(\ref{eq.3.1.24})$ to yield
\begin{equation*}
    \bra{0}\bra{0^m}\bra{i}U_A\ket{0}\ket{0^m}\ket{j} = \frac{1}{s}\sum\limits_{l,l'} A_{\nu(j,l),j}\delta_{i,\nu(j,l)}\delta_{\nu(i,l'),j}=\frac{1}{s}A_{ij}.
\end{equation*}
Thus, $U_A=U_L^\dagger U_R$ is an $(s,m+1,0)$-block-encoding of $A$, as depicted in Figure \ref{Block-encoding}. The gate complexity is $O({\rm poly}\log N)$ from the Hadamard gates and swap operator.
\end{proof}

\begin{figure}[H]
\centerline{
\Qcircuit @C=1.0em @R=1em{
    \lstick{\ket{0}} & \qw & \qw & \qw  & \qw & \gate{R} & \qw \barrier[-2.6em]{3} & \qw & \qw & \qw & \meter & \rstick{\ket{0}} \qw\\
    \lstick{\ket{0^m}} & {/} \qw & \gate{H^{\otimes m}} & \multigate{1}{\mathcal{O}_{A_1}} & \multigate{2}{\mathcal{O}_{A_2}} & \qw & \multigate{2}{\mathcal{O}_{A_2}^\dagger} & \multigate{1}{SWAP} & \multigate{1}{\mathcal{O}_{A_1}^\dagger} & \gate{H^{\otimes m}} & \meter & \rstick{\ket{0^m}} \qw \\
    \lstick{\ket{j}} & {/} \qw & \qw & \ghost{\mathcal{O}_{A_1}} & \ghost{\mathcal{O}_{A_2}} & \qw & \ghost{\mathcal{O}_{A_2}^\dagger} & \ghost{SWAP} & \ghost{\mathcal{O}_{A_1}^\dagger} & \qw & \qw & \rstick{\frac{1}{s}\sum_i A_{ij}\ket{i}} \qw\\
    \lstick{\ket{0}} & {/} \qw & \qw & \qw & \ghost{\mathcal{O}_{A_2}} & \ctrl{-3} & \ghost{\mathcal{O}_{A_2}^\dagger} & \qw & \qw & \qw & \qw & \rstick{\ket{0}} \qw\\
}
}
\caption{A schematic circuit for the block encoding $U_A$ of $A$. $m=\log(s)$. The first part before the dashed line is to implement $U_R$, then implement $U_L^\dagger$.}
\label{Block-encoding}
\end{figure}
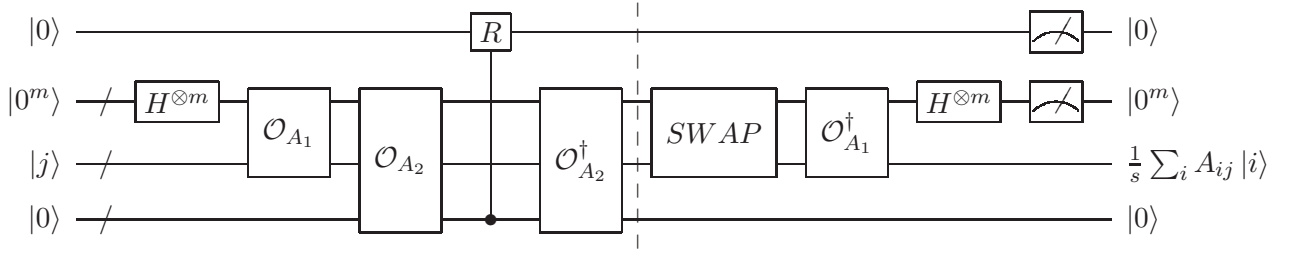

\begin{remark}
If the collocation points are properly enumerated, the matrix $A$ in (\ref{eq.2.3.20}) can be rearranged as a banded matrix due to the compactness of the radial function. In this sense, the nonzero entries in $A$ can be easily located, and the oracle $\mathcal{O}_{A_1}$ can be simplified.
\end{remark}

\begin{remark}
For simplicity, we assume the oracle $\mathcal{O}_{A_2}$ allows us to compute the exact value of the nonzero entries $A_{ij}$ of $A$. However, this can be relaxed. It can be seen from (\ref{eq.2.2.16}) and (\ref{eq.2.3.20}) that $A_{ij}$ depends on the pairwise distances $r_{ij}=||\bm x_i-\bm x_j||$ and the radial functions $\phi$ and $F_1$, $F_2$  that have explicit formulas from Lemma \ref{lemma.01}. In this sense, we can estimate $A_{ij}$ as follows. Employ the quantum distance estimation algorithm with suitable assumptions as in \cite{kerenidis2019q,  kerenidis2021quantum} to approximate the pairwise distances with $\ket{i,j}\ket{0}\mapsto\ket{i,j}\ket{r_{ij}}$. Then define oracles controlled by the index $i,j$ to evaluate the functions $\phi$, $F_1$ and $F_2$. And call the oracles to perform $\ket{i,j}\ket{r_{ij}}\ket{0} \mapsto \ket{i,j}\ket{r_{ij}}\ket{A_{ij}}$. Finally, return the register of  $\ket{r_{ij}}$ to $\ket{0}$ yielding $\ket{i,j}\ket{0} \mapsto \ket{i,j}\ket{A_{ij}}$.
\end{remark}

As for the quantum state $\ket{\bm b} := \bm b/||\bm b||$ of the right hand side of the linear system (\ref{eq.2.3.17}), we assume an oracle $\mathcal{O}_{\bm b}$ to prepare it like previous quantum algorithms \cite{harrow2009quantum, childs2017quantum, childs2021high}, i.e.,
\begin{equation}\label{eq.3.1.25}
    \mathcal{O}_{\bm b} \ket{0} = \ket{\bm b}.
\end{equation}

Finally, we consider block encodings of an initial- and final-state Hamiltonians that will be used for adiabatic evolution in the next subsection.

\begin{theorem} \label{th.06}
Let
\begin{equation}\label{eq.3.1.26}
    H_0 = \left[
    \begin{array}{cc}
        0 & P^{\perp}_{\bm b} \\
        P^{\perp}_{\bm b} & 0
    \end{array}
    \right], \quad
    H_1 = \left[
    \begin{array}{cc}
        0 & A P^{\perp}_{\bm b} \\
        P^{\perp}_{\bm b} A & 0
    \end{array}
    \right]
\end{equation}
with the projector $P^{\perp}_{\bm b}=I_n-\ket{\bm b}\bra{\bm b}$. Then
$P^{\perp}_{\bm b}$ has a $(1,1,0)$-block-encoding $U_{P^{\perp}_{\bm b}}$ with two uses of $\mathcal{O}_{\bm b}$. We can implement a $(1,1,0)$-block-encoding $U_{H_0}$ of $H_0$ with $O(1)$ uses of $\mathcal{O}_{\bm b}$. Besides,
we can implement an $(s,m+3,0)$-block-encoding $U_{H_1}$ of $H_1$ using $O(1)$ queries to $U_{A}$ and $\mathcal{O}_{\bm b}$, where $m=\log(s)$.
\end{theorem}
\begin{proof}
Note that $H_0$ and $H_1$ are of dimension $2N$, that is, they are $(n+1)$-qubit operators.
A $(1,1,0)$-block-encoding of $P^{\perp}_{\bm b}$ can be constructed as
\begin{equation*}
U_{P^{\perp}_{\bm b}} = \left[
    \begin{array}{cc}
        P^{\perp}_{\bm b} & \ket{\bm b}\bra{\bm b} \\
        \ket{\bm b}\bra{\bm b} & P^{\perp}_{\bm b}
    \end{array}
    \right]
    = \left(I_1\otimes \mathcal{O}_{\bm b}\right) \left( \sigma_x \otimes \ket{\bm 0}\bra{\bm 0} + I_1 \otimes \left(I_n-\ket{\bm 0}\bra{\bm 0}\right) \right) \left(I_1\otimes \mathcal{O}_{\bm b}\right)
\end{equation*}
with two uses of $\mathcal{O}_{\bm b}$, where $\sigma_x$ is the Pauli-$X$ matrix. Since $H_0 = \sigma_x \otimes P^{\perp}_{\bm b}$, it is easy to verify that $H_0$ has a $(1,1,0)$-block-encoding $(SWAP\otimes I_n)(\sigma_x \otimes U_{P^{\perp}_{\bm b}})(SWAP\otimes I_n)$ satisfying
\begin{equation*}
    \begin{split}
        &\left(\bra{0}\otimes I_1\otimes I_n\right) (SWAP\otimes I_n)(\sigma_x \otimes U_{P^{\perp}_{\bm b}})(SWAP\otimes I_n) \left(\ket{0}\otimes I_1\otimes I_n\right) \\
        = & \left(I_1\otimes\bra{0}\otimes I_n\right) (\sigma_x \otimes U_{P^{\perp}_{\bm b}}) \left(I_1\otimes \ket{0} \otimes I_n\right) \\
        = & \sigma_x \otimes \left(\bra{0}\otimes I_n\right) U_{P^{\perp}_{\bm b}} \left( \ket{0} \otimes I_n\right) \\
        = & \sigma_x \otimes P^{\perp}_{\bm b} = H_0
    \end{split}
\end{equation*}
where $SWAP$ is a $2$-qubit unitary such that $SWAP\ket{0}\otimes I_1=I_1\otimes\ket{0}$.

$H_1$ can be factored as
\begin{equation*}
    H_1 = \left[
    \begin{array}{cc}
        I_n & 0 \\
        0 & P^{\perp}_{\bm b}
    \end{array}
    \right]
    \left[
    \begin{array}{cc}
        0 & A \\
        A & 0
    \end{array}
    \right]
    \left[
    \begin{array}{cc}
        I_n & 0 \\
        0 & P^{\perp}_{\bm b}
    \end{array}
    \right].
\end{equation*}
The first and third parts in the decomposition are the same and have a $(1,1,0)$-block-encoding
\begin{equation*}
    \left[
    \begin{array}{cccc}
        I_n & 0 & 0 & 0 \\
        0 & P^{\perp}_{\bm b} & \ket{\bm b}\bra{\bm b} & 0 \\
        0 & \ket{\bm b}\bra{\bm b} & P^{\perp}_{\bm b} & 0 \\
        0 & 0 & 0 & I_n
    \end{array}
    \right]
    = \left[
    \begin{array}{ccc}
        I_n &  &   \\
         & U_{P^{\perp}_{\bm b}} &  \\
         &  & I_n
    \end{array}
    \right].
\end{equation*}
The second part is $\sigma_x \otimes A$. Since $A$ has an $(s,m+1,0)$-block-encoding $U_A$ from Theorem \ref{th.05}, it is easy to verify that
\begin{equation*}
    \begin{split}
        &\left(\bra{0^{m+1}}\otimes I_1\otimes I_n\right) (SWAP\otimes I_n)(\sigma_x \otimes U_A)(SWAP\otimes I_n) \left(\ket{0^{m+1}}\otimes I_1\otimes I_n\right) \\
        = & \left(I_1\otimes\bra{0^{m+1}}\otimes I_n\right) (\sigma_x \otimes U_A) \left(I_1\otimes \ket{0^{m+1}} \otimes I_n\right) \\
        = & \sigma_x \otimes \left(\bra{0^{m+1}}\otimes I_n\right) U_A \left( \ket{0^{m+1}} \otimes I_n\right) = \frac{1}{s}\left(\sigma_x \otimes A\right)
    \end{split}
\end{equation*}
where $SWAP$ is an $m+2$-qubit unitary such that $SWAP\ket{0^{m+1}}\otimes I_1=I_1\otimes\ket{0^{m+1}}$. As a result, $\sigma_x \otimes A$ can be embedded in an $(s,m+1,0)$-block-encoding $(SWAP\otimes I_n)(\sigma_x \otimes U_A)(SWAP\otimes I_n)$.
By the product of block-encoded matrices in \cite[Lemma 53]{gilyen2018quantum}, we can construct an $(s,m+3,0)$-block-encoding $U_{H_1}$ of $H_1$ using $O(1)$ queries to $U_{A}$ and $\mathcal{O}_{\bm b}$.
\end{proof}

\subsection{QLSA based on filtering}
By means of the technique of eigenstate filtering developed by Lin and Tong \cite{lin2020optimal}, we solve the preconditioned linear system (\ref{eq.2.3.17}) with the output of an $\epsilon_L$-approximation of $\ket{\bm c} := \bm c/||\bm c||$.
The strategy is evolving adiabatic quantum computing (AQC) for proper time to prepare a quantum state with a low accuracy to the target $\ket{\bm c}$; then applying a filtering function to filter out the unwanted information to finally get an $\epsilon_L$-close $\ket{\bm c}$, which achieves the near optimal complexity.

Firstly, based on AQC prepare a quantum state with a low accuracy to $\ket{\bm c}$. As in \cite{subacsi2019quantum,an2022quantum,lin2020optimal}, define the family of Hamiltonians
\begin{equation}\label{eq.3.2.27}
    H(f) = (1-f)H_0 + f H_1,
\end{equation}
with $f\in[0,1]$ a scheduling function, and $H_0$, $H_1$ the initial- and final-state Hamiltonians defined in (\ref{eq.3.1.26}). $H(f)$ acts on a Hilbert space of dimension $2N$. From \cite{subacsi2019quantum,lin2020optimal}, the nullspace of $H(f)$ can be spanned by $\ket{1}\ket{\bm b}$ and $\ket{0}\ket{c(f)}$, where $\ket{0}\ket{c(f)}$ is a time-dependent $0$-eigenstate of $H(f)$; besides, the spectral gap separating $0$ from other eigenvalues of $H(f)$ can be lower bounded by $\Delta_\ast(f) \equiv 1-f+f/\kappa$ with $\kappa$ the condition number of $A$. If we initialize our quantum computer in $\ket{0}\ket{c(0)}$, as $f$ is increased from $0$ to $1$, the ground state of $H(f)$ continuously changes from $\ket{0}\ket{c(0)}=\ket{0}\ket{\bm b}$ to $\ket{0}\ket{c(1)}=\ket{0}\ket{\bm c}$. Since the Hamiltonian $H(f)$ does not allow for transitions between the two ground states, the adiabatic path will stay along $\ket{0}\ket{c(f)}$ as long as the initial state is $\ket{0}\ket{c(0)}$ \cite{subacsi2019quantum,lin2020optimal}.
Note that $H_0$ and $H_1$ have null space spanned by $\{\ket{0}\ket{\bm b},\ket{1}\ket{\bm b}\}$ and $\{\ket{0}\ket{\bm c}, \ket{1}\ket{\bm b}\}$ respectively, and the spectral gap of $H_1$ is $1/\kappa$.

A good choice of the scheduling function $f$ is an important issue, because it determines the runtime of AQC. Various scheduling functions have been studied to improve the time complexity of AQC \cite{an2022quantum, subacsi2019quantum}. Here we exploit the time-optimal scheduling function proposed in \cite{an2022quantum} as
\begin{equation*}
    f(v) = \frac{\kappa}{1-\kappa}\left[ 1 - \left(1+v(\kappa^{p-1}-1)^{\frac{1}{1-p}}\right) \right]
\end{equation*}
for $1<p<2$, where $f: [0,1] \rightarrow [0,1]$, $f(0)=0$ and $f(1)=1$. Instead of using this scheduling function to evolve AQC for $O(\kappa/\epsilon_L)$ time to achieve $\epsilon_L$-close $\ket{\bm c}$ \cite{an2022quantum}, we run AQC for $O(\kappa)$ time to circumvent $O(\epsilon_L)$ query complexity. The output state is denoted by $\ket{\varphi}$, $O(1)$-close to $\ket{\bm c}$.

In conclusion, $\ket{\varphi}$ can be generated as follows. Consider the adiabatic evolution
\begin{equation}\label{eq.3.2.28}
    \frac{1}{T} {\mathbf i} \partial_{v}\ket{\psi_T(v)} = H\left(f(v)\right)\ket{\psi_T(v)}, \quad \ket{\psi_T(0)} = \ket{0}\ket{\bm b},
\end{equation}
where $v=t/T \in [0,1]$, and the parameter $T$ is called the runtime of AQC.
Using the block encodings of $H_0$ and $H_1$ from Theorem \ref{th.06}, construct a block encoding of $H(f)$ for the time-dependent Hamiltonian simulation \cite{low2018hamiltonian}. Run the adiabatic evolution for $T=O(\kappa)$ time to get the state $\ket{\varphi}$ with $O(1)$ accuracy to $\ket{\bm c}$, which requires $\widetilde{O}\left(s\kappa\log(s\kappa)\right)$ query complexity to oracles $\mathcal{O}_{A_1}, \mathcal{O}_{A_2}, \mathcal{O}_{\bm b}$ \cite{an2022quantum}. Note that $\ket{\varphi}$ can be decomposed as
\begin{equation}\label{eq.3.2.29}
    \ket{\varphi} = \mu_0\ket{0}\ket{\bm c} + \mu_1\ket{1}\ket{\bm b} + \sqrt{1-\mu_0^2-\mu_1^2}\ket{\perp}
\end{equation}
with $\mu_0=\Omega(1)$ a constant and $\ket{\perp}$ orthogonal to the nullspace $\{\ket{0}\ket{\bm c}, \ket{1}\ket{\bm b}\}$ of $H_1$; besides, since the AQC ensures $\bra{1}\braket{\bm b|\psi_T(v)}=0$ for all evolution time, $\mu_1$ should come entirely from the error of the Hamiltonian simulation.

Then we apply a filtering function to filter out the unwanted information from all other eigenstates of $\ket{\varphi}$, and obtain an $\epsilon_L$-approximation of the target state $\ket{\bm c}$.
By means of the Chebyshev polynomials, the eigenstate filtering function is defined as a $2\ell$-degree polynomial
\begin{equation}\label{eq.3.2.30}
    R_\ell \left(x; \Delta_\ast\right) := \frac{T_\ell (-1+2\frac{x^2-\Delta_\ast^2}{1-\Delta_\ast^2})}{T_\ell (-1+2\frac{-\Delta_\ast^2}{1-\Delta_\ast^2})},
\end{equation}
where $T_\ell$ is the $\ell$th Chebyshev polynomial of the first kind, and $\Delta_\ast$ denotes the spectral gap of a Hamiltonian in general. The nice properties of $R_\ell \left(x; \Delta_\ast\right)$ are shown as follows.

\begin{theorem}{\rm \cite{saad2003iterative,lin2020optimal}} \label{th.07}
{\rm (I)} $R_\ell(x,\Delta_\ast)$ solves the minimax problem
$$
\min\limits_{p(x)\in\mathbb{P}_{2\ell}, p(0)=1} \max\limits_{x\in D_{\Delta_\ast}} |p(x)|,
$$
where $\mathbb{P}_{2\ell}$ represents the polynomials of degree-$2\ell$ and $D_{\Delta_\ast} := [-1,-\Delta_\ast]\cup[\Delta_\ast, 1]$. \\
{\rm (II)} $|R_\ell \left(x; \Delta_\ast\right)|\leq 2e^{-\sqrt{2} \ell \Delta_\ast}$
for all $x\in D_{\Delta_\ast}$ and $0<\Delta_\ast\leq1/\sqrt{12}$. \\
{\rm (III)} $R_\ell(0;\Delta_\ast)=1$. $|R_\ell \left(x; \Delta_\ast\right)|\leq 0$ for all $|x|\leq 1$.
\end{theorem}

If the degree in the polynomial $R_\ell$ takes $\ell=\Omega(s\kappa\log(1/\epsilon_L))$, by multiplying $R_\ell(H_1/s;1/(s\kappa))$ to $\ket{\varphi}$,
we find that the $\ket{\perp}$ component of $\ket{\varphi}$ can be filtered out in the sense of $\epsilon_L$ accuracy as
\begin{equation*}
    \begin{split}
        R_\ell(H_1/s;1/(s\kappa))\ket{\varphi}
        & = R_\ell(H_1/s;1/(s\kappa))\left(\mu_0\ket{0}\ket{\bm c} + \mu_1\ket{1}\ket{\bm b} +\sqrt{1-\mu_0^2-\mu_1^2}\ket{\perp}\right) \\
        & = R_\ell(0;1/(s\kappa))\left(\mu_0\ket{0}\ket{\bm c} + \mu_1\ket{1}\ket{\bm b}\right) + \sqrt{1-\mu_0^2-\mu_1^2} R_\ell(H_1/s;1/(s\kappa))\ket{\perp} \\
        & = \mu_0\ket{0}\ket{\bm c} + \mu_1\ket{1}\ket{\bm b} + O(2e^{-\sqrt{2}\ell/(s\kappa)}) \\
        & = \mu_0\ket{0}\ket{\bm c} + \mu_1\ket{1}\ket{\bm b} + O(\epsilon_L).
    \end{split}
\end{equation*}
The second equality holds since $\ket{0}\ket{\bm c}$ and $\ket{1}\ket{\bm b}$ correspond to the $0$-eigenvalue of $H_1$ and $R_\ell(0;1/(s\kappa))=1$. The third equality is due to the fact that the other eigenvalues of $H_1/s$ lie in the range $D_{1/(s\kappa)}\equiv[-1,-1/(s\kappa)]\cup[1/(s\kappa),1]$, and by Theorem \ref{th.07} we have
\begin{equation*}
    ||R_\ell(H_1/s;1/(s\kappa))\ket{\perp}|| \leq 2e^{-\sqrt{2}\ell/(s\kappa)}.
\end{equation*}
In order to make $2e^{-\sqrt{2}\ell/(s\kappa)} \leq \epsilon_L$ it suffices to choose $\ell=\Omega(s\kappa\log(1/\epsilon_L))$,
and thus the last equality holds.

From Theorem \ref{th.06}, we can construct an $(s,m+3,0)$-block-encoding $U_{H_1}$ of $H_1$. Quantum eigenstate transformation via QSP then allows us to implement a $(1,m+4, 0)$-block-encoding $U_{R_\ell}$ of $R_\ell(H_1/s;1/(s\kappa))$ with $2\ell$ uses of $U_{H_1}$, $U_{H_1}^\dagger$ based on Theorem \ref{lemma.05}. Apply the unitary $U_{R_\ell}$ to the state $\ket{0^{m+4}}\ket{\varphi}$ to get
\begin{equation}\label{eq.5.2.42}
    \begin{split}
    U_{R_\ell} \ket{0^{m+4}}\ket{\varphi}
    & =\ket{0^{m+4}}\left(R_\ell(H_1/s;1/(s\kappa))\ket{\varphi} \right) + \ket{\perp} \\
    & = \ket{0^{m+4}}\left( \mu_0\ket{0}\ket{\bm c} + \mu_1\ket{1}\ket{\bm b} + \epsilon_L \right) + \ket{\perp}
    \end{split}
\end{equation}
with $(\bra{0^{m+4}}\otimes I_{n+1})\ket{\perp}=0$.

Finally, measure the ancilla qubits with outcome $\ket{0^{m+5}}$ to obtain a quantum state  $\epsilon_L$-close to $\ket{\bm c}$, with probability $|\mu_0|^2=\Omega(1)$. For a high success probability, $O(1)$ repetitions are sufficient.
Figure \ref{Linear_system} shows the quantum circuit of QLSA based on filtering to solve the preconditioned linear system $A\bm c =\bm b$ in (\ref{eq.2.3.17}). Putting all together, the total query complexity to output $\ket{\bm c}$ up to accuracy $\epsilon_L$ with high success probability is
\begin{equation}\label{eq.3.2.32}
\widetilde{O}\left(s\kappa\log(1/\epsilon_L)\right),
\end{equation}
and the time complexity is the same up to a ${\rm poly}\log(N)$ factor.

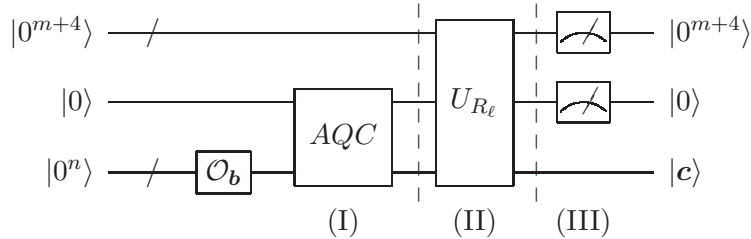
\begin{figure}[H]
    \centerline{
        \Qcircuit @C=1.5em @R=1em{
            \lstick{\ket{0^{m+4}}} & {/} \qw & \qw & \qw \barrier[-1.9em]{2} &\multigate{2}{U_{R_\ell}} \barrier[-1.6em]{2} & \meter & \rstick{\ket{0^{m+4}}} \qw \\
            \lstick{\ket{0}} & \qw & \qw & \multigate{1}{AQC} & \ghost{U_{R_\ell}} & \meter & \rstick{\ket{0}} \qw \\
            \lstick{\ket{0^n}} & {/} \qw & \gate{\mathcal{O}_{\bm b}} & \ghost{AQC} & \ghost{U_{R_\ell}} & \qw & \rstick{\ket{\bm c}} \qw \\
            &  &  & \rm{(I)} & \rm{(II)} & \rm{(III)} & \\
        }
    }
    \caption{Quantum circuit of filtering-based QLSA solving $A\bm c=\bm b$ in three stages. (I) Implement AQC for $O(\kappa)$ time to prepare $\ket{\varphi}$ in (\ref{eq.3.2.29}). (II) Apply the $(1,m+4, 0)$-block-encoding $U_{R_\ell}$. (III) Measure the first and second registers in states $\ket{0^{m+4}}$ and $\ket{0}$ to get an $\epsilon_L$-close $\ket{\bm c}$.}
    \label{Linear_system}
\end{figure}

\subsection{Numerical solution}
By means of the prepared state $\epsilon_L$-close to $\ket{\bm c}$, we can approximate the quantum state that encodes the numerical solution of the Poisson equation $\bar u(\bm x)$ in (\ref{eq.2.2.9}), i.e.,
\begin{equation*}
\bar u(\bm x) = -\sum\limits_{j=1}^{N_\mathcal{I}}c_j\Delta\Phi_\delta(\bm x-\bm x_j) + \sum\limits_{j=N_\mathcal{I}+1}^Nc_j\Phi_\delta(\bm x-\bm x_j)
\end{equation*}
on the set of collocation points $\mathcal{X}=\{\bm x_1,\dots,\bm x_N\}$. Define an $N\times N$ matrix
\begin{equation*}
M^\diamond := \left[
\begin{array}{cccccc}
    -\Delta\Phi_\delta(\bm x_1 - \bm x_1) & \cdots & -\Delta\Phi_\delta(\bm x_1 - \bm x_{N_{\mathcal I}}) & \Phi_\delta(\bm x_1 - \bm x_{N_{\mathcal I + 1}}) & \cdots & \Phi_\delta(\bm x_1 - \bm x_N) \\
    -\Delta\Phi_\delta(\bm x_2 - \bm x_1) & \cdots & -\Delta\Phi_\delta(\bm x_2 - \bm x_{N_{\mathcal I}}) & \Phi_\delta(\bm x_2 - \bm x_{N_{\mathcal I + 1}}) & \cdots & \Phi_\delta(\bm x_2 - \bm x_N) \\
    \vdots & \cdots & \vdots & \vdots & \cdots & \vdots \\
    -\Delta\Phi_\delta(\bm x_N - \bm x_1) & \cdots & -\Delta\Phi_\delta(\bm x_N - \bm x_{N_{\mathcal I}}) & \Phi_\delta(\bm x_N - \bm x_{N_{\mathcal I + 1}}) & \cdots & \Phi_\delta(\bm x_N - \bm x_N) \\
\end{array}
\right].
\end{equation*}
The numerical solution $\bar {\bm u} = [\bar u(\bm x_1),\dots,\bar u(\bm x_N)]^T$ on the set $\mathcal X$ has the representation
\begin{equation*}
    \bar {\bm u}=M^{\diamond}\bm c^\diamond = M^\diamond\mathcal{P}\bm c,
\end{equation*}
where $\bm c^\diamond=[c_1,\dots,c_N]^T$ in (\ref{eq.2.2.13}), and $\bm c^\diamond = \mathcal{P}\bm c$ with $\mathcal{P}={\rm diag}(\delta^2,\dots,\delta^2,1,\dots,1)$ a diagonal matrix from (\ref{eq.2.3.18}). By the expressions of $\Phi_\delta$ and $\Delta\Phi_\delta$ from (\ref{eq.2.1.7}) and Remark \ref{rem.01}, $M^\diamond$ can be further represented in the abbreviated notation as
\begin{equation*}
    M^\diamond = \left[ -\Delta\Phi_\delta(\bm x_i-\bm x_j), \Phi_\delta(\bm x_i - \bm x_j)\right] = \delta^{-d}\left[-\delta^{-2}F_1(r_{ij}/\delta), \phi(r_{ij}/\delta)\right]
\end{equation*}
analogous to that of $A^\diamond$ in (\ref{eq.2.2.15}), with $r_{ij}=||\bm x_i-\bm x_j||$. Note that $M^\diamond$ is also $s$-sparse depending on the pairwise distances in $\mathcal{X}$ and the support radius $\delta$.
Furthermore, define the $N\times N$ matrix
\begin{equation}\label{eq.3.3.33}
M := M^\diamond\mathcal{P} = \delta^{-d}\left[-F_1(r_{ij}/\delta), \phi(r_{ij}/\delta)\right]
\end{equation}
in the abbreviated notation, which is also $s$-sparse. The quantum state of $\bar {\bm u}$ is given by
\begin{equation}\label{eq.3.3.34}
\ket{\bar {\bm u}} := \frac{\bar{\bm u}}{||\bar{\bm u}||} = \frac{M\ket{\bm c}}{||M\ket{\bm c}||}.
\end{equation}
From \cite[Theorem 2.1]{duan2013stability}, some theoretical results with respect to $M$ can be concluded as follows, and the detailed proof is given in Appendix B.

\begin{theorem} \label{th.08}
Let $q\in (0,1]$ be the separation distance of $\mathcal X$ defined in (\ref{eq.2.2}). For the $s$-sparse matrix $M$ in (\ref{eq.3.3.33}), we have $||M\bm\xi||\geq C q^{-d} (q/\delta)^{2\tau}||\bm\xi||$ for any $\bm\xi\in\mathbb{R}^{N}$, and the condition number ${\rm cond}(M)\leq C s (\delta/q)^{2\tau-d}$, where $C$ is a positive constant independent of $\mathcal X$.
\end{theorem}

To approximate $\ket{\bar{\bm u}}$, consider block encodings to embed $M$ into a larger unitary. Since $-\Delta \Phi$ and $\Phi$ are positive definite from Remark \ref{rem.01},  the maximum element in absolute value of $M$ is $C \delta^{-d}$ by Theorem \ref{th.01}, with $C=\max\{-\Delta\Phi(\bm 0), \Phi(\bm 0)\}$ a positive constant. Define a rescaled version of $M$ as
\begin{equation}\label{eq.3.3.35}
\hat M := M/(C\delta^{-d})
\end{equation}
such that each element of $\hat M$ has absolute value not greater than 1 and the block embedding can be achieved. Exploit the standard matrix dilation method to extend $\hat{M}$ as a Hermitian matrix of dimension $2N$, i.e.,
\begin{equation*}
    \mathcal M := \left[
\begin{array}{cc}
    0 & \hat M \\
    \hat M^\dagger & 0 \\
\end{array}
\right].
\end{equation*}
Similar to the technique of encoding $A$ as a block inside a larger unitary in Theorem \ref{th.05}, assume oracles that have access to the nonzero entries of $\mathcal M$. Then we can perform an $(s,m+1,0)$-block-encoding $U_{\mathcal M}$ of $\mathcal M$ with $O(1)$ uses of the oracles in time $O({\rm poly}\log N)$, where $m=\log(s)$. Apply the unitary $U_{\mathcal M}$ to the state $\ket{0^{m+1}}\ket{1}\ket{\bm c}$ to yield
\begin{equation*}
    \begin{split}
    U_{\mathcal M}\ket{0^{m+1}}\ket{1}\ket{\bm c}
    & = \ket{0^{m+1}} \left(\frac{1}{s}\mathcal M \ket{1}\ket{\bm c}\right) + \ket{\perp} \\
    & = \ket{0^{m+1}}\ket{0} \left(\frac{1}{s}\hat M \ket{\bm c}\right) + \ket{\perp} \\
    \end{split}
\end{equation*}
with $(\bra{0^{m+1}}\otimes I_{n+1})\ket{\perp}=0$. Finally, measure the ancilla qubits to get the quantum state $\ket{\bar{\bm u}}$ in (\ref{eq.3.3.34}), with probability $\left(||\hat M \ket{\bm c}|| /s\right)^2$. Figure \ref{Numerical_solution} shows the related quantum circuit.
\begin{figure}[H]
    \centerline{
        \Qcircuit @C=1.5em @R=1em{
            \lstick{\ket{0^{m+1}}} & {/} \qw & \qw  &\multigate{2}{U_{\mathcal M}} & \meter & \rstick{\ket{0^{m+1}}} \qw \\
            \lstick{\ket{0}} & \qw & \gate{\sigma_x} & \ghost{U_{\mathcal M}} & \meter & \rstick{\ket{0}} \qw \\
            \lstick{\ket{\bm c}} & {/} \qw & \qw & \ghost{U_{\mathcal M}} & \qw & \rstick{\ket{\bar{\bm u}}} \qw \\
        }
    }
    \caption{Quantum circuit to generate $\ket{\bar{\bm u}}$, where $\sigma_x$ is the Pauli-$X$ matrix.}
    \label{Numerical_solution}
\end{figure}
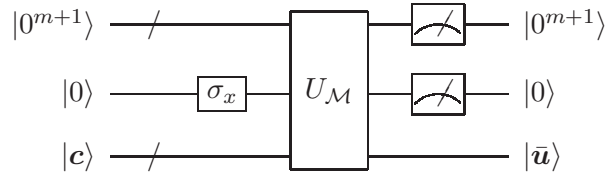

Theorem \ref{th.08} ensures that $||M \ket{\bm c}||\geq q^{-d} (q/\delta)^{2\tau}$. We can immediately derive that
\begin{equation*}
    \frac{1}{s}||\hat M \ket{\bm c}||
    = \frac{1}{C \delta^{-d} s} ||M \ket{\bm c}||
    \geq \frac{1}{C s} \left(\frac{q}{\delta}\right)^{2\tau-d}.
\end{equation*}
Using amplitude amplification \cite{brassard2002quantum}, $O\left(s(\delta/q)^{2\tau-d}\right)$ repetitions are sufficient for success. Taking into account $\widetilde{O}\left(s\kappa\log(1/\epsilon_L) {\rm poly}(\log N)\right)$ time complexity from (\ref{eq.3.2.32}) to prepare $\ket{\bm c}$ with error $\epsilon_L$, and $O({\rm poly}\log N)$ time to implement $U_{\mathcal M}$, the overall time complexity to approximate $\ket{\bar{\bm u}}$ with high probability is
\begin{equation}\label{eq.3.3.36}
    \widetilde{O}\left((\delta/q)^{2\tau-d}s^2\kappa\log(1/\epsilon_L) {\rm poly}(\log N)\right).
\end{equation}

\subsection{Error analysis}
Combining the derived results above, we perform an overall error analysis and complexity estimate for the quantum RBF method solving the Poisson problem that is defined in Problem \ref{problem 1} in Section 2. With the assumption that the collocation points have quasi-uniform distribution satisfying (\ref{eq.2.3}) in Problem \ref{problem 1}, the number of collocation points $N$ will grow as
\begin{equation}\label{eq.3.4.37}
    N = O\left(h^{-d}\right),
\end{equation}
where $h$ is the mesh norm defined in (\ref{eq.2.1}). And the separation distance $q$ can be roughly regard as $q\approx h$. The goal is to output an $\epsilon$-close quantum state that encodes the solution of the Poisson equation on the set $\mathcal X$ of collocation points, where $\epsilon$ is the desired precision as algorithm input.

There are two types of errors to be considered in the analysis: the RBF error and the quantum algorithm error. For analyzing simplicity, let $\ket{\bm u} := \bm u/||\bm u||$ be the exact quantum state proportional to $\bm u=[u(\bm x_1),\dots,u(\bm x_N)]^T$ of the exact solution of the Poisson equation on $\mathcal X$. Let $\ket{\widetilde{\bm u}}$ be the actual output state from quantum procedure. Let $\ket{\widetilde{\bm c}}$ be the actual output state from QLSA in Subsection 3.2. Then $\ket{\widetilde{\bm u}}$ can be represented as $\ket{\widetilde {\bm u}} = M\ket{\widetilde{\bm c}}/||M\ket{\widetilde{\bm c}}||$, analogous to the definition of the numerical solution $\ket{\bar{\bm u}}$ on $\mathcal X$ in (\ref{eq.3.3.34}).
The total error of approximating $\ket{\bm u}$ by $\ket{\widetilde{\bm u}}$ is bounded by
\begin{equation}\label{eq.3.4.38}
||\ket{\bm u}-\ket{\widetilde{\bm u}}|| \leq ||\ket{\bm u}-\ket{\bar{\bm u}}|| + ||\ket{\bar{\bm u}}-\ket{\widetilde{\bm u}}||.
\end{equation}

Firstly, concerning the RBF error $||\ket{\bm u}-\ket{\bar{\bm u}}||$, we have
\begin{equation}\label{eq.3.4.39}
||\ket{\bm u}-\ket{\bar{\bm u}}||
\leq 2 \frac{||\bm u-\bar{\bm u}||}{||{\bm u}||}
\leq C \frac{||u-\bar{u}||_{L_2(\Omega)}}{||u||_{L_2(\Omega)}}
\leq C \frac{||u||_{H^\tau(\Omega)}}{||u||_{L_2(\Omega)}} \left(\frac{h^{1-2/\tau}}{\delta}\right)^{\tau}
\end{equation}
where $C>0$ is a constant.
Note that the boldface denotes vectors, italics denote functions, and kets denote quantum states.
The first inequality can be easily derived by the property of triangle inequality of vector norm. The second is based on the $L_2$-norm discretization \cite{freeman2023discretizing} from which we have
\begin{equation*}
    c_1 ||u||^2_{L_2(\Omega)} \leq \frac{1}{N} ||\bm u||^2 \leq c_2 ||u||^2_{L_2(\Omega)}
\end{equation*}
with $c_1 \leq c_2$ positive constants. And the third inequality applies Theorem \ref{th.02}.
For a given Poisson problem, $||u||_{H^\tau(\Omega)}$ and $||u||_{L_2(\Omega)}$ should be constants independent of the data points. To ensure $||\ket{\bm u}-\ket{\bar{\bm u}}||\leq \epsilon/2$, we can choose
\begin{equation}\label{eq.3.4.40}
\left(h^{1-2/\tau}/\delta\right)^\tau=O\left(\epsilon\right).
\end{equation}

Next, take into account the error $||\ket{\bar{\bm u}}-\ket{\widetilde{\bm u}}||$ induced by inaccuracies from QLSA solving the linear system (\ref{eq.2.3.17}). Since $||\ket{\bm c}-\ket{\widetilde{\bm c}}||\leq \epsilon_L$ from Subsection 3.2, it is easy to verify that
\begin{equation}\label{eq.3.4.41}
\begin{split}
& ||\ket{\bar{\bm u}}-\ket{\widetilde{\bm u}}||
= \bigg|\bigg|\frac{M\ket{\bm c}}{||M\ket{\bm c}||} - \frac{M\ket{\widetilde{\bm c}}}{||M\ket{\widetilde{\bm c}}||} \bigg|\bigg| \\
& \leq \bigg|\bigg|\frac{M\ket{\bm c}}{||M\ket{\bm c}||} - \frac{M\ket{\widetilde{\bm c}}}{||M\ket{\bm c}||} \bigg|\bigg| + \bigg|\bigg|\frac{M\ket{\widetilde{\bm c}}}{||M\ket{\bm c}||} - \frac{M\ket{\widetilde{\bm c}}}{||M\ket{\widetilde{\bm c}}||} \bigg|\bigg| \\
& \leq 2\epsilon_L \frac{||M||}{||M\ket{\bm c}||}  \leq 2 {\rm cond}(M) \cdot \epsilon_L,
\end{split}
\end{equation}
where $M$ is defined in (\ref{eq.3.3.33}) and the condition number ${\rm cond}(M)$ can be bounded from Theorem \ref{th.08}.
To make $||\ket{\bar{\bm u}}-\ket{\widetilde{\bm u}}||\leq\epsilon/2$, it suffices to choose the linear system error
\begin{equation}\label{eq.3.4.42}
\epsilon_L=O\left(\epsilon/{\rm cond}(M)\right).
\end{equation}

In conclusion, with the relationship (\ref{eq.3.4.40}) and the linear system error $\epsilon_L$ in (\ref{eq.3.4.42}), we can approximate $\ket{\bm u}$ by $\ket{\widetilde{\bm u}}$ up to accuracy $\epsilon$. The runtime is of the form (\ref{eq.3.3.36}). Putting all together, the main result that solves Problem \ref{problem 1} is summarized as follows.

\begin{theorem}\label{th.main}
Assume the support radius $\delta=C h^{1-\beta/\tau}$ in $(0,1]$, where $C$ is a positive constant, $\beta>2$ is a constant determining the RBF error in (\ref{eq.3.4.39}), $h$ is the sufficiently small mesh norm defined by (\ref{eq.2.1}), and $\tau=d/2+k+1/2$ indicates the order of smoothness of the Sobolev space generated by $\Phi$ in (\ref{eq.2.2.10}).
Given Problem \ref{problem 1} in Section 2, there exists an RBF quantum algorithm that outputs an $\epsilon$-close state $\ket{\bm u}$ encoding the solution of the Poisson equation (\ref{eq.2.2.8}) at the collocation points, with runtime
\begin{equation}\label{eq.3.4.43}
    \widetilde{O}\left(\epsilon^{-\frac{\beta}{\beta-2}\left(4+\frac{d}{\tau}\right)}\right).
\end{equation}
\end{theorem}
\begin{proof}
Substitute the condition number $\kappa$ and the sparsity $s$ by their upper bounds in Theorem \ref{th.03}. And replace $N$, $\epsilon_L$ with the formulas (\ref{eq.3.4.37}) and (\ref{eq.3.4.42}). The overall runtime from (\ref{eq.3.3.36}) becomes
\begin{equation}\label{eq.3.4.44}
\begin{split}
& \widetilde{O}\left((\delta/q)^{2\tau-d}s^2\kappa\log(1/\epsilon_L){\rm poly}\log N\right) \\
= & \widetilde{O}\left(\left(\delta/q\right)^{4\tau-2d} (1+\delta/q)^{3d} \log\left({\rm cond}(M)/\epsilon \right) {\rm poly}\log\left(1/h^{d}\right)\right) \\
= & \widetilde{O}\left((\delta/h)^{4\tau+d} \log(1/\epsilon) \log(1/h) \right) \\
=& \widetilde{O}\left(h^{-\beta\left(4+\frac{d}{\tau}\right)}\right) =\widetilde{O}\left(\epsilon^{-\frac{\beta}{\beta-2}\left(4+\frac{d}{\tau}\right)}\right).
\end{split}
\end{equation}
The first equality applies the substitution of $\kappa$, $s$, $N$ and $\epsilon_L$. The second equality is built on the approximations $q\approx h$, $(1+\delta/q)^d \approx (\delta/h)^d$ due to the quasi-uniform data distribution, and the upper bound of ${\rm cond}(M)$ in Theorem \ref{th.08}. And the last two equality uses the assumption $\delta=C h^{1-\beta/\tau}$ with $\beta>2$ and the relationship $h^{\beta-2}=O(\epsilon)$ from (\ref{eq.3.4.40}).
\end{proof}

If we employ the best classical method, the conjugate gradient method, to solve the resulting linear system and then compute the solution of the Poisson equation on $\mathcal X$, the runtime is
\begin{equation}\label{eq.3.4.45}
\begin{split}
O\left(Ns\sqrt{\kappa}\log(1/\epsilon)\right)
& = \widetilde{O}\left(h^{-d} (1+\delta/q)^{1.5d} (\delta/q)^{\tau-0.5d}  \log(1/\epsilon)  \right) \\
& = \widetilde{O}\left( h^{-d} (\delta/h)^{\tau+d} \log(1/\epsilon) \right) \\
& = \widetilde{O}\left(h^{-\beta\left(1+\frac{d}{\tau}+\frac{d}{\beta}\right)} \right)
= \widetilde{O}\left(\epsilon^{-\frac{\beta}{\beta-2}\left(1+\frac{d}{\tau} + \frac{d}{\beta} \right)} \right)
\end{split}
\end{equation}
where the equality is deduced like that of Theorem \ref{th.main}.

Table \ref{table.2} compares the time complexity of the quantum and classical methods. By the runtime of quantum and classical methods in (\ref{eq.3.4.44}) and (\ref{eq.3.4.45}), we can see that when $d>3\beta$ the quantum algorithm can achieve a polynomial speedup. Since $\beta>2$ for convergence, the quantum advantage appears for high-dimensional problems with $d>6$.

\begin{table}[H]
    \centering
    \renewcommand\arraystretch{1.5}
    \caption{\label{table.2} Time complexity of the quantum and classical methods. $\delta=C h^{1-\beta/\tau}$, with $C>0$, $\beta>2$ constants and $\tau=d/2+k+1/2$. Q-advantage: the case where our quantum algorithm shows an advantage over the classical counterpart.}
    \begin{tabular}{l c c c}
        \toprule
        \textbf{Algorithm} & \textbf{Classical} & \textbf{Quantum}  &  \textbf{Q-advantage} \\
        \midrule
        \textbf{Complexity} & $\widetilde{O}\left(\epsilon^{-\frac{\beta}{\beta-2}\left(1+\frac{d}{\tau} + \frac{d}{\beta} \right)} \right)$ & $\widetilde{O}\left(\epsilon^{-\frac{\beta}{\beta-2}\left(4+\frac{d}{\tau}\right)}\right)$ & $d>3\beta$ \\
        \bottomrule
    \end{tabular}
\end{table}

To be more clear, take $\beta=3,4$ for examples. That is, the support radius $\delta$ takes $C h^{1-3/\tau}$ and $C h^{1-4/\tau}$ such that the RBF error in (\ref{eq.3.4.39}) scales as $h$ and $h^2$ respectively. Table \ref{table.3} lists the related runtime of the classical and quantum algorithms. The following observations can be made.
\begin{itemize}
    \item[$\bullet$] The quantum speedup attained with respect to the precision $\epsilon$ is only at most polynomial if the space dimension $d$ is fixed. Moreover, in a low dimension, the runtime of the quantum algorithm can actually be worse than that of the classical method.

    \item[$\bullet$] The runtime of the quantum algorithm does not grow exponentially in $d$.
    For high-dimensional problems, the quantum speedup can be very significant.

    \item[$\bullet$] As $\beta$ increases, the support radius $\delta=C h^{1-\beta/\tau}$ increases. For fixed mesh norm $h$, this results in better accuracy from (\ref{eq.3.4.39}), but worse sparsity. That is, a tradeoff principle exists.
    The quantum advantage appears in the case of higher space dimension, since the quantum advantage strongly relies on the sparsity of the matrix which can be seen from (\ref{eq.3.3.36}).
\end{itemize}

\begin{table}[H]
    \centering
    \renewcommand\arraystretch{1.5}
    \caption{\label{table.3} Time complexity of the quantum and classical methods when $\beta=3,4$.}
    \begin{tabular}{l l l c}
        \toprule
        \textbf{Algorithm} & \textbf{Classical} & \textbf{Quantum} &  \textbf{Q-advantage} \\
        \midrule
        \textbf{$\beta=3$} & $\widetilde{O}\left(\epsilon^{-3-\frac{3d}{\tau} - d} \right)$ & $\widetilde{O}\left(\epsilon^{-12-\frac{3d}{\tau}}\right)$ & $d>9$ \\
        \textbf{$\beta=4$} & $\widetilde{O}\left(\epsilon^{-2 - \frac{2d}{\tau} - \frac{d}{2}} \right)$ & $\widetilde{O}\left(\epsilon^{-8-\frac{2d}{\tau}}\right)$ & $d>12$ \\
        \bottomrule
    \end{tabular}
\end{table}

\section{Conclusion}
In this paper, we present a quantum RBF method to deal with the Poisson problem. We apply the symmetric collocation method based on CSRBFs to discretize the Poisson equation as a linear system. Using the QLSA based on filtering, we solve the preconditioned linear system. Furthermore, we employ the technique of block encoding to get a quantum state encoding the solution of the Poisson equation at the collocation points. Compared with the best classical method, the quantum algorithm can achieve a polynomial speedup for high-dimensional problems. The algorithm may have possible applications for high-dimensional Poisson problems with irregular geometry. The results can also be extended to general second order elliptic boundary value problems.

It is an interesting question of whether the quantum RBF method can achieve an exponential speedup. For future research, we may consider the problem of other types of PDEs using the RBF method. Another interesting direction is to consider the multiscale collocation method which may overcome the tradeoff principle.

\section*{Acknowledgements}
This work was supported by the National Key Research and Development
Program of China under Grant No.\ 2021YFA1000600, and the National
Natural Science Foundation of China under Grant No.\ 11571265. %, and NSFC/RGC Joint Research Scheme No.\  12061160462.

\begin{appendices}
\setcounter{equation}{0}
\renewcommand{\theequation}{A\arabic{equation}}

\section{Proof of Lemma 2}
\renewcommand{\thelemma}{2}
\begin{lemma}\label{lemma.A.2}
Let $\Phi(\bm x)=\phi(r)$ with support radius $1$ and $r=||\bm x||$ for $\bm x=(x_1,\dots,x_d)\in\mathbb{R}^d$. Suppose $\Phi\in C^4(\mathbb{R}^d)$. Denote by $\Delta$ the Laplace operator $\Delta=\sum_{j=1}^{d}\partial^2/\partial x_j^2$.
Then we have
\begin{equation}\label{eq.A.A1}
    \Delta\Phi(\bm x) = \phi''(r)+\frac{d-1}{r}\phi'(r),
\end{equation}
and
\begin{equation}\label{eq.A.A2}
    \Delta^2\Phi(\bm x) = \phi^{(4)}(r)+\frac{2(d-1)}{r}\phi'''(r)+\frac{(d-1)(d-3)}{r^2}\phi''(r)-\frac{(d-1)(d-3)}{r^3}\phi'(r)
\end{equation}
provided that $r\neq0$ or when $r=0$ these two formulas are well defined,
where $\Delta^2$ represents the double Laplacian.
\end{lemma}
\begin{proof}
Firstly, we prove the formula (\ref{eq.A.A1}) for second-order derivatives of the radial function $\Phi$ of $d$ variables. For the $j$th variable $x_j$ ($j=1,\dots,d$) of $\bm x$, the chain rule implies
\begin{equation*}
    \begin{aligned}
        \frac{\partial}{\partial x_j}\phi(r) &=\frac{d\phi(r)}{dr}\frac{\partial }{\partial x_j}r(x_1,\dots,x_d)\\
        &=\frac{d\phi(r)}{dr}\frac{x_j}{\sqrt{x_1^2+\cdots+x_d^2}}
        = \frac{x_j}{r}\phi'(r)
    \end{aligned}
\end{equation*}
since $r=||\bm x||=\sqrt{x_1^2+\cdots+x_d^2}$. The generic second-order derivatives are given by
\begin{equation*}
    \frac{\partial}{\partial x_j^2}\phi(r)
    =\frac{\partial}{\partial x_j}\left(\frac{x_j}{r}\phi'(r)\right)
    = \frac{x_j^2}{r^2}\phi''(r)+\frac{r^2-x_j^2}{r^3}\phi'(r)
\end{equation*}
for $j=1,\dots,d$.
Hence, the Laplacian is of the form
\begin{equation*}
    \begin{aligned}
        \Delta\Phi(\bm x) &= \Delta\phi(||\bm x||) =\left(\frac{\partial^2}{\partial x_1^2}+\dots+\frac{\partial^2}{\partial x_d^2}\right)\phi(r)\\
        &= \phi''(r)+\frac{d-1}{r}\phi'(r),
    \end{aligned}
\end{equation*}
which completes the proof of (\ref{eq.A.A1}).

We then prove the formula (\ref{eq.A.A2}). Applying the Laplacian to the formula (\ref{eq.A.A1}) yields
\begin{equation*}
    \begin{aligned}
        \Delta^2\Phi(\bm x) & =\left(\frac{\partial^2}{\partial x_1^2}+\dots+\frac{\partial^2}{\partial x_d^2}\right)\left(\phi''(r)+\frac{d-1}{r}\phi'(r)\right).\\
    \end{aligned}
\end{equation*}
The chain rule implies
\begin{equation*}
    \begin{aligned}
        &\frac{\partial}{\partial x_j}\left(\phi''(r)+\frac{d-1}{r}\phi'(r)\right)\\
        &= \frac{x_j}{r}\phi'''(r) + \frac{\partial}{\partial x_j}\left(\frac{d-1}{r}\phi'(r)\right) \\
        & = \frac{x_j}{r}\phi'''(r) + \frac{(d-1)x_j}{r^2}\phi''(r) - \frac{(d-1)x_j}{r^3}\phi'(r),
    \end{aligned}
\end{equation*}
and
\begin{equation*}
    \begin{aligned}
        & \frac{\partial^2}{\partial x_j^2}\left(\phi''(r)+\frac{d-1}{r}\phi'(r)\right)\\
        & =\frac{\partial}{\partial x_j}\left(\frac{x_j}{r}\phi'''(r) + \frac{(d-1)x_j}{r^2}\phi''(r) - \frac{(d-1)x_j}{r^3}\phi'(r)\right)\\
        &= \frac{\partial}{\partial x_j}\left(\frac{x_j}{r}\phi'''(r)\right) + \frac{\partial}{\partial x_j}\left(\frac{(d-1)x_j}{r^2}\phi''(r)\right) - \frac{\partial}{\partial x_j}\left(\frac{(d-1)x_j}{r^3}\phi'(r)\right) \\
        & = \frac{x_j^2}{r^2}\phi^{(4)}(r) + \frac{r^2+(d-2)x_j^2}{r^3}\phi'''(r)+
        \frac{(d-1)(r^2-3x_j^2)}{r^4}\phi''(r) -\frac{(d-1)(r^2-3x_j^2)}{r^5}\phi'(r)
    \end{aligned}
\end{equation*}
for $j=1,\dots,d$. Thus, we have
\begin{equation*}
    \begin{aligned}
        \Delta^2\Phi(\bm x) & =\left(\frac{\partial^2}{\partial x_1^2}+\dots+\frac{\partial^2}{\partial x_d^2}\right)\left(\phi''(r)+\frac{d-1}{r}\phi'(r)\right)\\
        & = \phi^{(4)}(r)+\frac{2(d-1)}{r}\phi'''(r)+\frac{(d-1)(d-3)}{r^2}\phi''(r)-\frac{(d-1)(d-3)}{r^3}\phi'(r),
    \end{aligned}
\end{equation*}
which finishes the proof of (\ref{eq.A.A2}).
\end{proof}

\section{Proof of Theorem 8}
We demonstrate Theorem $\ref{th.08}$ based on \cite[Theorem 2.1]{duan2013stability}.
Take the substitution $\bm z = \bm x/ \delta$ and let the Laplace operator act on $\bm z$, the matrix $M$ in (\ref{eq.3.3.33}) can be further expressed in the abbreviated notation as
\begin{equation}\label{eq.B.A3}
    M=\delta^{-d}[-\Delta\Phi(\bm z_i-\bm z_j),\Phi(\bm z_i-\bm z_j)],
\end{equation}
and the separation distance in terms of $\bm z$ is $q/\delta$. It is easy to derive the following theorem.

\renewcommand{\thetheorem}{8}
\begin{theorem}
    Let $q\in (0,1]$ be the separation distance of $\mathcal X$ defined in (\ref{eq.2.2}). For the $s$-sparse matrix $M$ defined in (\ref{eq.3.3.33}), we have $||M\bm\xi||\geq C q^{-d} (q/\delta)^{2\tau}||\bm\xi||$ for any $\bm\xi\in\mathbb{R}^{N}$, and the condition number ${\rm cond}(M)\leq C s (\delta/q)^{2\tau-d}$, where $C$ is a positive constant independent of the data set $\mathcal X$.
\end{theorem}
\begin{proof}
    Since $\Phi$ takes the form of Wendland's functions, the decay condition (\ref{eq.2.1.6}) is satisfied. By \cite[Theorem 2.1]{duan2013stability} and (\ref{eq.B.A3}), we have
    \begin{equation*}
        ||M\bm\xi||\geq C \delta^{-d} (q/\delta)^{-d} (q/\delta)^{2\tau}||\bm\xi|| = C q^{-d} (q/\delta)^{2\tau}||\bm\xi||.
    \end{equation*}
    That is, $\sqrt{\lambda_{\min}(M^{T}M)} \geq C q^{-d} (q/\delta)^{2\tau}$. Besides since $M$ is $s$-sparse with respect to the row and column, and the largest value in $M$ is $C \delta^{-d}$  with $C=\max\{-\Delta\Phi(\bm 0), \Phi(\bm 0)\}$ a positive constant, we have
    $||M|| \leq \sqrt{||M||_1 ||M||_\infty} \leq C s \delta^{-d}$. Thus, we can derive the upper bound of the condition number of $M$.
\end{proof}
\end{appendices}

\end{document}